\RequirePackage{fix-cm}
\documentclass[12pt]{article}
\usepackage{fixltx2e}[2014/09/29]
\usepackage[style=numeric-comp,maxnames=10,backend=biber,%
            hyperref=true]{biblatex}
\usepackage[hidelinks]{hyperref}
\usepackage{amsmath,amssymb,amsthm,thmtools}
\usepackage[utf8]{inputenc}
\usepackage{mathtools}
\usepackage{enumitem,booktabs}
\mathtoolsset{centercolon,mathic}
\usepackage{listings}
\usepackage{xcolor,xspace}
\usepackage{tikz-cd}
\usepackage{microtype}

\usetikzlibrary{positioning}
\usetikzlibrary{arrows.meta}
\usetikzlibrary{decorations.pathreplacing}
\tikzset{thick red/.style={thick,color=darkred}}
\tikzset{edgelabel/.style = {midway,font=\small}}
\tikzset{slopedlabel/.style = {edgelabel,sloped}}
\definecolor{darkred}{RGB}{170,0,0}
\tikzcdset{arrow style=tikz}

\makeatletter
\protected\def\tikz@nonactivecolon{\ifmmode\mathrel{\mathop\ordinarycolon}\else:\fi} 
\makeatother

\definecolor{keywordcolor}{rgb}{0.7, 0.1, 0.1}   % red
\definecolor{tacticcolor}{rgb}{0.1, 0.2, 0.6}    % blue
\definecolor{commentcolor}{rgb}{0.4, 0.4, 0.4}   % grey
\definecolor{symbolcolor}{rgb}{0.0, 0.1, 0.6}    % blue
\definecolor{sortcolor}{rgb}{0.1, 0.5, 0.1}      % green
\definecolor{DarkBlue}{rgb}{0.03, 0.21, 0.48}    % Dark Blue (R9, G53 B122)

\mathchardef\texthyphen="2D
\newcommand*{\susp}[1]{{\Sigma}{#1}}

\newcommand*{\rev}[1]{{#1}^{-1}}
\newcommand*{\bN}{{\mathbb{N}}}
\newcommand*{\bF}{{\mathbb{F}}}
\newcommand*{\bZ}{{\mathbb{Z}}}

\newcommand*{\bS}{{\mathbb{S}}}
\newcommand*{\RP}{{\mathbb{R}\mathrm{P}}}

\newcommand*{\pt}{{\mathrm{pt}}}

\newcommand*{\dblslash}{\mathbin{/\kern-3pt/}}
\renewcommand*{\equiv}{\mathrel{\simeq}}
\DeclarePairedDelimiter\angled{\langle}{\rangle} % angle brackets
\DeclarePairedDelimiter\trunc{\lVert}{\rVert} % truncation
\DeclarePairedDelimiter\tr{\lvert}{\rvert} % introduction for truncation
\DeclareMathOperator\Decat{Decat}
\DeclareMathOperator\Disc{Disc}

\DeclareMathOperator\Aut{Aut}
\DeclareMathOperator\BAut{BAut}
\DeclareMathOperator\ap{ap}
\DeclareMathOperator\B{B}
\newcommand*\BS{\mathop{BS}}    %or: \mathrm{B}\Sigma} ?
\newcommand*\Sym{S} % Symmetric group (or \Sigma ?)
\newcommand*\BZ{\mathop{\mathrm{B}\bZ}}
\DeclareMathOperator\id{id}
\DeclareMathOperator\im{im}

\DeclareMathOperator\fib{fib}
\DeclareMathOperator\Fin{Fin}
\DeclareMathOperator\istrunc{istrunc}
\DeclareMathOperator\iscontr{iscontr}

\DeclareMathOperator\isconn{isconn}

\newcommand*{\Set}{\mathrm{Set}}
\newcommand*{\Type}{\mathrm{Type}}
\newcommand*{\GType}{\mathrm{GType}} % (k-tuply groupal) n-types
\newcommand*{\Grp}{\mathrm{\texthyphen Group}} % higher groups
\newcommand*{\Group}{\mathrm{Group}}
\newcommand*{\AbGroup}{\mathrm{AbGroup}}
\newcommand*{\iDecat}{\mathrm{\infty\texthyphen\mathrm{Decat}}}
\newcommand*{\iDisc}{\mathrm{\infty\texthyphen\mathrm{Disc}}}
\renewcommand*{\phi}{\varphi}
\newcommand*{\blank}{{-}}
\newcommand*{\leanref}[1]{{\normalfont[\texttt{\color{DarkBlue}#1}]}}
\newcommand*{\ditto}{--- \raisebox{-0.5ex}{''} ---}

\declaretheorem{theorem}
\declaretheorem{lemma}
\declaretheorem{proposition}
\declaretheorem{corollary}

\declaretheorem[style=definition]{definition}

\title{Higher Groups in Homotopy Type Theory}
\date{\today}
\author{Ulrik Buchholtz, Floris van Doorn and Egbert Rijke}
\begin{document}

\maketitle

\begin{abstract}
  We present a development of the theory of higher groups, including
  infinity groups and connective spectra, in homotopy type theory. An
  infinity group is simply the loops in a pointed, connected type,
  where the group structure comes from the structure inherent in the
  identity types of Martin-L{\"o}f type theory. We investigate
  ordinary groups from this viewpoint, as well as higher dimensional
  groups and groups that can be delooped more than once. A major result
  is the stabilization theorem, which states that if an $n$-type can be
  delooped $n+2$ times, then it is an infinite loop type. Most of the
  results have been formalized in the Lean proof assistant.
\end{abstract}

\section{Introduction}
\label{sec:introduction}

\begin{table*}
  \caption{\label{tab:periodic}Periodic table of $k$-tuply groupal $n$-groupoids.}
\makebox[\textwidth][c]{
  \begin{tabular}{clllll} \toprule
    $k\setminus n$ & $0$ & $1$ & $2$ & $\cdots$ & $\infty$ \\
    \midrule
    $0$ & pointed set & pointed groupoid & pointed $2$-groupoid & $\cdots$ & pointed $\infty$-groupoid \\
    $1$ & group & $2$-group & $3$-group & $\cdots$ & $\infty$-group \\
    $2$ & abelian group & braided $2$-group & braided $3$-group & $\cdots$ & braided $\infty$-group \\
    $3$ & \ditto & symmetric $2$-group & sylleptic $3$-group & $\cdots$ & sylleptic $\infty$-group \\
    $4$ & \ditto & \ditto & symmetric $3$-group & $\cdots$ & ?? $\infty$-group \\
    $\vdots$ & \mbox{}\quad$\vdots$ & \mbox{}\quad$\vdots$ & \mbox{}\quad$\vdots$ & $\ddots$ & \mbox{}\quad$\vdots$ \\
    $\omega$ & \ditto & \ditto & \ditto & $\cdots$ & connective spectrum \\
    \bottomrule
  \end{tabular}%
  }
\end{table*}

The homotopy hypothesis is the statement that homotopy $n$-types (topological spaces with trivial homotopy groups above level $n$) correspond to $n$-groupoids for $n\in\mathbb{N}\cup\{\infty\}$ via the fundamental $\infty$-groupoid construction.
In Grothendieck's original version in \emph{Pursuing
  Stacks}~\cite{Grothendieck1983} this was a conjecture
about a particular model of $\infty$-groupoids.
It is also a theorem for many particular models of $\infty$-groupoids,
for example the Kan simplicial sets,
but it is now mostly taken to be a property \emph{defining}
$\infty$-groupoids up to equivalence.

In this paper, we investigate the homotopy hypothesis in the context of homotopy type theory (HoTT). HoTT refers to the homotopical interpretation of
Martin-L{\"o}f's dependent type theory \cite{AwodeyWarren2009,Voevodsky06}. 
In this homotopical interpretation, every type-theoretical construction corresponds to a
homotopy-invariant construction on spaces. 

In HoTT, every type has a path space given by the identity type. For a
pointed type we can construct the loop space, which has the structure
of an $\infty$-group. Moreover, if the type is \emph{truncated}, then
we can retreive the usual notion of groups, 2-groups and higher
groups. This allows us to define a higher group internally in the
language of type theory as a type that is the loop space of a pointed connected type, its delooping.

We also investigate groups that can be delooped more than once, which gives $n$-groups with additional coherences. The full family of groups we consider is in \autoref{tab:periodic}, which we will explain in detail in \autoref{sec:higher-groups}.

Our approach is additionally validated by the corresponding observation in
$\infty$-topos theory, where it is a theorem that the $\infty$-category of
pointed, connected objects in $\mathcal X$ is equivalent to the
$\infty$-category of higher group objects in $\mathcal X$, for any
$\infty$-topos $\mathcal X$~\cite[Lemma~7.2.2.11(1)]{LurieHTT}.

We have formalized most of our results in the HoTT library~\cite{vDvRB2017HoTTLean} of the Lean Theorem Prover~\cite{Moura2015}. 
The formalized results can be found in the file
\url{https://github.com/cmu-phil/Spectral/blob/master/higher_groups.hlean}.
We will indicate the major formalized results in this paper by
referring to the name in the formalization inside square brackets.
For more information about the formalization, see \autoref{sec:formalization}.

We are indebted to Michael Shulman for writing a blog
post~\cite{Shulman-classifying} on classifying spaces from a univalent
perspective.

\section{Preliminaries}
\label{sec:preliminaries}

In this paper we will work in the type theory of the HoTT book~\cite{TheBook}, although all arguments will also hold in a cubical type theory, such as \cite{CCHM2016,chtt}. In this section we briefly introduce the concepts we need for the rest of the paper.

The type theory contains dependent function types $(x : A) \to B(x)$, which are more traditionally denoted as $\Pi_{x : A}B(x)$ and dependent pair types $(x : A) \times B(x)$, which are traditionally denoted as $\Sigma_{x : A}B(x)$. We choose to use this Agda-inspired notation because we often deal with deeply nested dependent sum types.

Within a type $A$ we have the identity type or path type ${=}_A:A\to A\to \Type$. We have various operations on paths, such as concatenation $p\cdot q$ and inversion $\rev{p}$ of paths. The functorial action of a function $f:A\to B$ on a path $p:a_1=_A a_2$ is denoted $\ap_f(p):f(a_1)=f(a_2)$. The constant path is denoted $1_a:a=a$.

A type $A$ can be $n$-truncated, denoted $\istrunc_n A$, which is
defined by recursion on $n : \bN_{-2} := \mathbb{Z}_{\ge-2}$:
\begin{align*}
  \istrunc_{-2} A & := \iscontr A := (a : A) \times \bigl((x : A) \to (a = x)\bigr) \\
  \istrunc_{n+1} A & := (x\;y : A) \to \istrunc_n (x=y)
\end{align*}
For any type $A$ we write $\trunc{A}_n$ for its $n$-truncation, i.e.,~$\trunc{A}_n$ is an $n$-truncated type equipped with a map $\tr{\blank}_n:A\to\trunc{A}_n$ such that for any $n$-truncated type $B$ the precomposition map
\begin{equation*}
(\trunc{A}_n\to B)\to (A\to B)
\end{equation*}
is an equivalence. Then we define being $n$-connected as $\isconn_n A := \iscontr \trunc A_n$. Properties of truncations and connected maps are established in Chapter 7 of \cite{TheBook}.

The type of pointed types is
$\Type_\pt := (A : \Type) \times (\pt : A)$. The type of $n$-truncated types is
$\Type^{\le n} := (A : \Type) \times \istrunc_n A$ and for $n$-connected types it is
$\Type^{>n} := (A : \Type) \times \isconn_n A$. We will combine these notations as needed.

Given $A:\Type_\pt$ we define the \emph{loop space} $\Omega
A:=(\pt=_A\pt)$, which is pointed with basepoint $1_\pt$. The
\emph{homotopy groups} of $A$ are defined to be $\pi_k A:=
\trunc{\Omega^k A}_0$. These are group in the usual sense when
$k\ge1$, with neutral element $\tr{1}$ and group
operation induced by path concatenation.

Given $A,B:\Type_\pt$ the type of \emph{pointed maps} from $A$ to $B$
is $(A\to_\pt B):=(f : A \to B)\times (f(\pt)=_B\pt)$. Given $f:A\to_\pt
B$ we write $f(a):B$ for the first projection and $f_0:f(\pt)=\pt$ for
the second projection. The \emph{fiber} of a pointed map is defined by
$\fib(f):=(a:A)\times (f(a)=_B\pt)$, which is pointed with basepoint $(\pt,f_0)$.

In HoTT we can use \emph{higher inductive types} to construct Eilenberg-MacLane spaces $K(G,n)$~\cite{FinsterLicata2014}. For a group $G$ we define $K(G,1)$ as the following HIT.\\
\indent \texttt{HIT} $K(G,1) :=$ \\
\indent $\bullet\ \star : K(G,1)$; \\
\indent $\bullet\ p : G \to \star=\star$; \\
\indent $\bullet\ q : (g\;h : G) \to p(gh) = p(g) \cdot p(h)$; \\
\indent $\bullet\ \epsilon : \istrunc_1K(G,1)$. \\
(Using the univalent universe $\Type$, other direct definitions are
also possible, for instance, $K(G,1)$ is equivalent to the type of
small $G$-torsors.)
Let $\Sigma X$ denote the suspension of $X$, i.e., the homotopy pushout of $1\leftarrow X\to 1$. For an abelian group $A$ can now inductively define $K(A,n+1):=\trunc{\Sigma K(A,n)}_{n+1}$. Then we have the following result~\cite{FinsterLicata2014}.

\begin{theorem}\label{thm:EM-spaces}
  Let $G$ be a group and $n\geq 1$, and assume that $G$ is abelian
  when $n>1$. The space $K(G,n)$ is $(n-1)$-connected and
  $n$-truncated and there is a group isomorphism $\pi_nK(G,n) \equiv G$.
\end{theorem}

% possible other things to mention: 
% - loop-suspension adjunction

In some of our informal arguments we use the descent theorem for
pushouts,\footnote{Recall from~\cite[\S6.1.3]{LurieHTT}, following
  ideas from Charles Rezk, that we can
  \emph{define} the $\infty$-toposes among locally cartesian closed
  $\infty$-categories as those whose colimits are \emph{van Kampen},
  viz., satisfying descent.}
which states that for a commuting cube of types
\begin{equation}\label{eq:cube}
\begin{tikzcd}
& A_{11} \arrow[dl] \arrow[dr] \arrow[d] \\
A_{10} \arrow[d] & B_{11} \arrow[dl] \arrow[dr] & A_{01} \arrow[dl,crossing over] \arrow[d] \\
B_{10} \arrow[dr] & A_{00} \arrow[d] \arrow[from=ul,crossing over] & B_{01} \arrow[dl] \\
& B_{00},
\end{tikzcd}
\end{equation}
if the bottom square is a pushout and the vertical squares are pullbacks, then the top square is also a pushout. We will use the following slight generalization.

\begin{theorem}
Consider a commuting cube of types as in~\eqref{eq:cube}, and suppose the vertical squares are pullback squares.
Then the square
\begin{equation*}
\begin{tikzcd}
A_{10}\sqcup^{A_{11}} A_{01} \arrow[r] \arrow[d] & A_{00} \arrow[d] \\
B_{10}\sqcup^{B_{11}} B_{01} \arrow[r] & B_{00}
\end{tikzcd}
\end{equation*}
is a pullback square.
\end{theorem}

\begin{proof}
It suffices to show that the pullback 
\begin{equation*}
(B_{10}\sqcup^{B_{11}} B_{01})\times_{B_{00}}A_{00}
\end{equation*}
has the universal property of the pushout. This follows by the descent theorem, since by the pasting lemma for pullbacks we also have that the vertical squares in the cube
\begin{equation*}
\begin{tikzcd}
& A_{11} \arrow[dl] \arrow[dr] \arrow[d] \\
A_{10} \arrow[d] & B_{11} \arrow[dl] \arrow[dr] & A_{01} \arrow[dl,crossing over] \arrow[d] \\
B_{10} \arrow[dr] & (B_{10}\sqcup^{B_{11}} B_{01})\times_{B_{00}}A_{00} \arrow[d] \arrow[from=ul,crossing over] & B_{01} \arrow[dl] \\
& B_{10}\sqcup^{B_{11}} B_{01}
\end{tikzcd}
\end{equation*}
are pullback squares.
\end{proof}
In the formalization, arguments using descent are more conveniently
done via the equivalent principle captured formally as the flattening
lemma~\cite[\S6.12]{TheBook}.

\section{Higher groups}
\label{sec:higher-groups}

Recall that types in HoTT may be viewed as $\infty$-groupoids:
elements are objects, paths are morphisms, higher paths are higher
morphisms, etc.

It follows that \emph{pointed connected} types $A$ may be viewed as higher
groups, with \emph{carrier} $\Omega A := (\pt =_A \pt)$.
The neutral element is the identity path,
the group operation is given by path composition,
and higher paths witness the unit and associativity laws.
Of course, these higher paths are themselves subject to further laws,
etc., but the beauty of the type-theoretic definition is
that we don't have to worry about that:
all the (higher) laws follow from the rules of the identity types.

Writing $G$ for the carrier, it is common to write $BG$ for the pointed
connected type such that $G = \Omega BG$.
We call $BG$ the \emph{delooping} of $G$.
Let us write
\begin{align*}
  \infty\Grp
  &:= (G : \Type) \times (BG : \Type_\pt^{>0}) \times (G \equiv \Omega BG) \\
  &\phantom{:}\equiv (G : \Type_\pt) \times (BG : \Type_\pt^{>0}) \times (G \equiv_\pt \Omega BG) \\
  &\phantom{:}\equiv \Type_\pt^{>0}
\end{align*}
for the type of higher groups, or \emph{$\infty$-groups}.
Note that for $G:\infty\Grp$
we also have $G:\Type$ using the first projection as a coercion. Using
the last definition, this is the loop space map, and not the usual coercion!

We recover the ordinary set-level
groups by requiring that $G$ is a $0$-type, or equivalently, that $BG$
is a $1$-type. This leads us to introduce
\begin{align*}
  n\Grp
  &:= (G : \Type_\pt^{<n}) \times (BG : \Type_\pt^{>0})
    \times (G \equiv_\pt \Omega BG) \\
  &\phantom{:}\equiv \Type_\pt^{>0,\le n}
\end{align*}
for the type of \emph{groupal} (group-like) $(n-1)$-groupoids, also known
as \emph{$n$-groups}. For $G:1\Grp$
a set-level group, we have $BG = K(G,1)$.

For example, the integers $\bZ$ as an additive group are from this
perspective represented by their delooping $\BZ=\bS^1$, i.e., the circle.

Of course, double loop spaces are even better behaved than mere loop
spaces (e.g., they are commutative up to homotopy
by the Eckmann-Hilton argument~\cite[Theorem~2.1.6]{TheBook}).
Say a type $G$ is
\emph{$k$-tuply groupal} if we have a $k$-fold delooping,
$B^kG : \Type_\pt^{\ge k}$, such that $G = \Omega^kB^kG$.

Mixing the two directions, let us introduce the type
\begin{align*}
  (n,k)\GType
  &:= (G : \Type_\pt^{\le n}) \times (B^kG : \Type_\pt^{\ge k})
    \times (G \equiv_\pt \Omega^kB^kG) \\
  &\omit$\phantom{:}\equiv \Type_\pt^{\ge k,\le n+k}$\hfill\text{\leanref{GType\_equiv}}
\end{align*}
for the type of \emph{$k$-tuply groupal $n$-groupoids}.\footnote{This
  is called $n\Type_k$ in \cite{BaezDolan1998}, but here we give equal
  billing to $n$ and $k$,
  and we add the ``G'' to indicate group-structure.}
(We allow taking $n=\infty$ in which case the truncation requirement
is simply dropped. \leanref{InfGType\_equiv})

Note that $n\Grp = (n-1,1)\GType$. This shift in indexing is slightly
annoying, but we keep it to stay consistent with the literature.

Since there are forgetful maps
\begin{equation*}
(n,k+1)\GType \to (n,k)\GType
\end{equation*}
given by $B^{k+1}G\mapsto \Omega B^{k+1}G$
we can also allow $k$ to be infinite, $k=\omega$ by setting
\begin{align*}
(n,\omega)\GType & := \lim\nolimits_k{}(n,k)\GType \\
&\phantom{:}\equiv \bigl(B^{-}G : (k : \bN) \to \Type_\pt^{\ge k,\le n+k}\bigr) \\
  &\qquad \times \bigl((k : \bN) \to B^kG \equiv_\pt \Omega B^{k+1}G \bigr).
\end{align*}
In \autoref{sec:stabilization} we prove the stabilization theorem
(\autoref{thm:stabilization}), from which it follows that
$(n,\omega)\GType=(n,k)\GType$ for $k\geq n+2$.

When $(n,k)=(\infty,\omega)$, this is the type of stably groupal $\infty$-groups,
also known as \emph{connective spectra}. If we also relax the
connectivity requirement, we get the type of all spectra, and we can
think of a spectrum as a kind of $\infty$-groupoid with $k$-morphisms
for all $k\in\bZ$.

The class of higher groups is summarized in~\autoref{tab:periodic}.
We shall prove the correctness of the $n=0$ column in~\autoref{sec:n=0}.

\section{Elementary theory}
\label{sec:elementary-theory}

Given \emph{any} type of objects $A$, any $a : A$ has an
\emph{automorphism group} $\Aut_A a:=\Aut a:=(a=a)$ with
$\BAut a = \im(a : 1 \to A) = (x : A) \times \trunc{a=x}_{-1}$ (the
connected component of $A$ at $a$). Clearly, if $A$ is
$(n+1)$-truncated, then so is $\BAut a$ and so $\Aut a$ is
$n$-truncated, and hence an $(n+1)$-group.

Moving across the homotopy hypothesis, for every pointed type $(X,x)$
we have the \emph{fundamental $\infty$-group of $X$},
$\Pi_\infty(X,x):=\Aut x$. Its $(n-1)$-truncation (an instance of
decategorification, see \autoref{sec:stabilization}) is the
\emph{fundamental $n$-group of $X$}, $\Pi_n(X,x)$,
with corresponding delooping $\mathrm{B}\Pi_n(X,x) = \trunc{\BAut x}_n$.

If we take $A = \Set$, we get the usual symmetric groups
$\Sym_n := \Aut(\Fin n)$, where $\Fin n$ is a set with $n$
elements. (Note that $\BS_n = \BAut(\Fin n)$ is the type of all
$n$-element sets.)
We give further constructions related to ordinary groups
in \autoref{sec:perspectives}.

\subsection{Homomorphisms and conjugation}
\label{sec:homomorphisms}

A homomorphism between higher groups is any
function that can be suitably delooped.
For $G,H : (n,k)\GType$, we define
\begin{align*}
\hom_{(n,k)}(G,H) & := 
(h: G \to_\pt H)\times (B^k h: B^kG \to_\pt B^kH) \\
& \qquad \times (\Omega^k(B^k h) \sim_\pt h) \\
&\phantom{:}\equiv (B^k h: B^kG \to_\pt B^kH).
\end{align*}
For (connective) spectra we need
pointed maps between all the deloopings and pointed homotopies showing
they cohere.

Note that if $h,k : G \to H$ are homomorphisms between set-level
groups, then $h$ and $k$ are \emph{conjugate} if $Bh, Bk : BG \to_\pt BH$ are
\emph{freely} homotopic (i.e., equal as maps $BG \to BH$).

Also observe that $\pi_j(B^kG \to_\pt B^kH) \equiv
\trunc{B^kG \to_\pt \Omega^kB^kH}_0 \equiv
\trunc{\Sigma^jB^kG \to_\pt B^kH}_0 = 0$ for $j+k-1 \ge n+k$,
that is, for $j>n$, so
this suggests that $\hom_{(n,k)}(G,H)$ is $n$-truncated.
(The calculation verifies this for the identity component.)
To prove
this, we need to use an induction using the definition of
$n$-truncated. If $f:\hom_{(n,k)}(G,H)$, then its self-identity type is
equivalent to $\bigl(\alpha : (z : B^kG) \to (f\,z = f\,z)\bigr) \times
\bigl(\alpha\,\pt \cdot g_\pt = f_\pt\bigr)$. This type is no longer a
type of pointed maps, but rather a type of pointed sections of a
fibration of pointed types.
\begin{definition}
  If $X : \Type_\pt$ and $Y : X \to \Type_\pt$, then we introduce the
  type of \emph{pointed sections},
  \[
    (x : X) \to_\pt Y\,x
    := \bigl(s : (x : X) \to Y\,x\bigr)
    \times \bigl(s\,\pt = \pt\bigr).
  \]
  This type is itself pointed by the trivial section $\lambda x, \pt$.
\end{definition}
% Compared to is_trunc_ppi_of_is_conn we have
% k := n+2
% n := k+1
% hence n+3+k there is n+k for us
\begin{theorem}
  Let $X : \Type_\pt^{\ge k}$ be an $(k-1)$-connected, pointed type for
  some $k\ge0$, and
  let $Y : X \to \Type_\pt^{\le n+k}$ be a fibration of
  $(n+k)$-truncated, pointed types for some $n\ge -1$. Then the type
  of pointed sections, $(x : X) \to_\pt Y\,x$, is $n$-truncated.
  \leanref{is\_trunc\_ppi\_of\_is\_conn}
\end{theorem}
\begin{proof}
  The proof is by induction on $n$.

  For the base case $n=-1$ we have to show that the type of pointed
  sections is a mere proposition. Since it is pointed, it must in fact
  be contractible. The center of contraction is the trivial section $s_0$.
  If $s$ is another section, then we get a pointed homotopy from $s$
  to $s_0$ from the elimination principle for pointed, connected
  types~\cite[Lemma~7.5.7]{TheBook}, since the types $s\,x=s_0\,x$ are
  $(k-2)$-truncated.

  To show the result for $n+1$, taking the $n$ case as the induction
  hypothesis, it suffices to show for any pointed section $s$ that its
  self-identity type is $n$-truncated. But this type is equivalent to
  $(x : X) \to_\pt \Omega(Y\,x, s\,x)$, which is again a type of
  pointed sections, and here we can apply the induction hypothesis.
\end{proof}
\begin{corollary}
  Let $k\ge0$ and $n\ge-1$. If $X$ is $(k-1)$-connected, and $Y$ is
  $(n+k)$-truncated, then the type of pointed maps $X \to_\pt Y$ is
  $n$-truncated. In particular, $\hom_{(n,k)}(G,H)$ is an $n$-type for
  $G,H:(n,k)\GType$.
\end{corollary}
\begin{corollary}
  The type $(n,k)\GType$ is $(n+1)$-truncated.
  \leanref{is\_\allowbreak trunc\_\allowbreak GType}
\end{corollary}
\begin{proof}
  This follows immediately from the preceding corollary, as the type
  of equivalences $G \equiv H$ is a subtype of the homomorphisms from
  $G$ to $H$.
\end{proof}

If $k\ge n+2$ (so we're in the stable range), then $\hom_{(n,k)}(G,H)$
becomes a stably groupal $n$-groupoid. This generalizes the
fact that the homomorphisms between abelian groups form an abelian
group.

The automorphism group $\Aut G$ of a higher group $G:(n,k)\GType$ is in
$(n,1)\GType$.
This is equivalently the automorphism group of the pointed type $B^kG$.
But we can also forget the basepoint and consider the
automorphism group $\Aut^c G$ of $B^kG:\Type^{\ge k,\le n+k}$. This
now allows for (higher) conjugations. We define the \emph{generalized
  center} of $G$ to be $ZG := \Omega^k\!\Aut^c G : (n,k+1)\GType$
(generalizing the center of a set-level group,
see below in~\autoref{sec:center}).

\subsection{Group actions}
\label{sec:actions}

In this section we consider a fixed group $G : \GType$ with delooping
$BG$. An \emph{action} of $G$ on some object of type $A$ is simply
a function $X : BG \to A$. The object of the action is $X(\pt) : A$,
and it can be convenient to consider evaluation at $\pt : BG$ to be a
coercion from actions of type $A$ to $A$. To equip $a : A$ with a
$G$-action is to give an action $X : BG \to A$ with $X(\pt)=a$. The
\emph{trivial action} is the constant function at $a$. Clearly, an
action of $G$ on $a:A$ is the same as a homomorphism $G \to \Aut_A a$.

If $A$ is a universe of types, then we have actions on types $X : BG
\to \Type$. These $G$-types are thus simply types in the context of
$BG$. A map of $G$-types from $X$ to $Y$ is just a function $\alpha :
(z : BG) \to X(z) \to Y(z)$.

If $X$ is a $G$-type, then we can form the
\begin{description}
\item[invariants] $X^{hG} := (z : BG) \to X(z)$, also known as the
  \emph{homotopy fixed points}, and the
\item[coinvariants] $X_{hG} := (z : BG) \times X(z)$, which is also known as
  \emph{homotopy orbit space} or the \emph{homotopy quotient}
  $X \dblslash G$.
\end{description}
It is easy to see that these constructions are respectively the right and left
adjoints of the functor that sends a type $X$ to the trivial
$G$-action on $X$, $X^{\mathrm{triv}} : BG \to \Type$, which is just
the constant family at $X$. Indeed, the adjunctions are just the usual argument
swap and (un)currying equivalences, for $Y:\Type$,
\begin{align*}
  \hom(Y, X^{hG})
  &= X \to (z : BG) \to Y(z)
    \equiv (z : BG) \to X \to Y(z) \\
  &\equiv \hom(X^{\mathrm{triv}}, Y), \\
  \hom(X_{hG}, Y)
  &= \big((z : BG) \times X(z)\bigr) \to Y
    \equiv (z : BG) \to X(z) \to Y \\
  &\equiv \hom(X, Y^{\mathrm{triv}}).
\end{align*}
If we think of an action $X : BG\to \Type$
as a type-valued diagram on $BG$, this means that the homotopy fixed
points and the homotopy orbit space form the homotopy limit and
homotopy colimit of this diagram, respectively.

\begin{proposition}
  \label{prop:action-hom-pullback}
  Let $f : H \to G$ be a homomorphism of higher groups with delooping
  $Bf : BH \to_\pt BG$, and let $\alpha : \hom(X,Y)$ be a map of
  $G$-types. By composing with $f$ we can also view $X$ and $Y$ as
  $H$-types, in which case we get a homotopy pullback square:
  \[
    \begin{tikzcd}
      X_{hH} \ar[r]\ar[d] & Y_{hH}\ar[d] \\
      X_{hG} \ar[r]       & Y_{hG}.
    \end{tikzcd}
  \]
\end{proposition}
\begin{proof}
  The vertical maps are induced by $Bf$, and the horizontal maps are
  induced by $\alpha$. The homotopy pullback corner type $C$ is calculated as
  \begin{align*}
    C
    &\equiv (z : BG) \times (x : X\,z)
      \times (w : BH) \times (y : Y(Bf\,w)) \\
    &\qquad \times (z = Bf\,w) \times (y = \alpha\,z\,x) \\
    &\equiv (w : BH) \times (x : X(Bf\,w))
      = X_{hH},
  \end{align*}
  and under this equivalence the top and the left maps are the
  canonical ones.
\end{proof}

Every group $G$ carries two canonical actions on itself:
\begin{description}
\item[the right action] $G : BG \to \Type$, $G(x) = (\pt = x)$, and the
\item[the adjoint action] $G^{\mathrm{ad}} : BG \to \Type$,
  $G^{\mathrm{ad}}(x) = (x = x)$ (by conjugation).
\end{description}

We have $1 \dblslash G = BG$,
$G \dblslash G = 1$ and $G^{\mathrm{ad}} \dblslash G = LBG :=
(\bS^1 \to BG)$, the free loop space of $BG$. Recalling that $\BZ =
\bS^1$, we see that $G^{\mathrm{ad}} = (\BZ \to BG)$, i.e., the
conjugacy classes of homomorphisms from $\bZ$ to $G$. Since the
integers are the free (higher) group on one generator, this is just
the conjugacy classes of elements of $G$. But that is exactly what we
should get for the homotopy orbits of $G$ under the conjugation
action.

The above proposition has an interesting corollary:
\begin{corollary}
  \label{cor:hfiber-hom}
  If $f : H \to G$ is a homomorphism of higher groups,
  then $G \dblslash H$ is equivalent to the homotopy
  fiber of the delooping $Bf : BH \to_\pt BG$,
  where $H$ acts on $G$ via the $f$-induced right action.
\end{corollary}
\begin{proof}
  We apply Proposition~\ref{prop:action-hom-pullback} with $\alpha : G \to 1$
  being the canonical map from the right action of $G$ to the action
  of $G$ on the unit type. Then the square becomes:
  \[
    \begin{tikzcd}[baseline=(O.base)]
      G \dblslash H \ar[r]\ar[d] & BH\ar[d] \\
      |[alias=O]| 1 \ar[r]       & BG
    \end{tikzcd}\qedhere
  \]
\end{proof}

By definition, $BG$ classifies \emph{principal $G$-bundles}: pullbacks
of the right action of $G$. That is, a principal $G$-bundle over a
type $A$ is a family $F : X \to \Type$ represented by a map $\chi:A
\to BG$ such that $F(x) \equiv (\pt = \chi(x))$ for all $x : X$.

For example, for every higher group $G$ we have the corresponding Hopf
fibration $\Sigma G \to \Type$ represented by the map $\chi_H : \Sigma
G \to BG$ corresponding under the loop-suspension adjunction to the
identity map on $G$. (This particular fibration can be defined using
only the induced $H$-space structure on $G$.)

This perspective underlies the construction of the first and the third
named author of the real projective spaces in homotopy type
theory~\cite{realprojective}. The fiber sequences $\bS^0 \to \bS^n \to
\RP^n$ are principal bundles for the $2$-elements group
$\bS^0=\Sym_2$ with delooping $\BS_2\equiv \RP^\infty$,
the type of $2$-element types.

\subsection{Back to the center}
\label{sec:center}

We mentioned the generalized center above and claimed that it
generalized the usual notion of the center of a group.
Indeed, if $G:1\Grp$ is a set-level group, then an element of
$ZG$ corresponds to an element of $\Omega^2\BAut^cG$, or equivalently,
a map from the $2$-sphere $\bS^2$ to $\Type$
sending the basepoint to $BG$.
By the universal property of
$\bS^2$ as a HIT, this again corresponds to a homotopy from the
identity on $BG$ to itself, $c : (z : BG) \to z = z$.
This is precisely a homotopy fixed point of the adjoint action of
$G$ on itself, i.e., a central element.

\subsection{Equivariant homotopy theory}
\label{sec:equivariant}

Fix a group $G : \GType$. Suppose that $G$ is actually the (homotopy)
type of a topological group. Consider the
type $BG \to \Type$ of (small) \emph{types with a $G$-action}. Naively,
one might think that this represents $G$-equivariant homotopy types,
i.e., sufficiently nice\footnote{Sufficiently nice means the
  $G$-CW-spaces. The same homotopy category arises by taking all
  spaces with a $G$-action, but then the weak equivalences are the
  $G$-maps $f : X \to Y$ that induce weak equivalences on $H$-fixed
  point spaces $f^H : X^H \to Y^H$ for all closed subgroups $H$ of $G$.}
topological spaces with a $G$-action considered up to $G$-equivariant
homotopy equivalence. But this is not so.

By Elmendorf's theorem~\cite{Elmendorf1983}, this homotopy theory is
rather that of presheaves of (ordinary) homotopy types on the
\emph{orbit category} $\mathcal{O}_G$ of $G$. This is the full
subcategory of the category of $G$-spaces spanned by the homogeneous
spaces $G/H$, where $H$ ranges over the closed subgroups of $G$.

Inside the orbit category we find a copy of the group $G$, namely as
the endomorphisms of the object $G/1$ corresponding to the trivial
subgroup $1$. Hence, a $G$-equivariant homotopy type gives rise to
type with a $G$-action by restriction along the inclusion $BG
\hookrightarrow \mathcal{O}_G$. (Here we consider $BG$ as a (pointed
and connected) topological groupoid on one object.)

As remarked by Shulman~\cite{ShulmanEI}, when $G$ is a \emph{compact} Lie
group, then $\mathcal{O}_G$ is an inverse EI $\infty$-category, and
hence we know how to model type theory in the presheaf $\infty$-topos
over $\mathcal{O}_G$. And in certain simple cases we can even define
this model internally. For instance, if $G=\bZ/p\bZ$ is a cyclic group
of prime order, then a small $G$-equivariant type consists of a type with a
$G$-action, $X : BG \to \Type$ together with another type family $X^G : X^{hG}
\to \Type$, where $X^G$ gives for each homotopy fixed point a type of
proofs or ``special reasons'' why that point should be considered
fixed~\cite[7.6]{ShulmanEI}. Hence the total space of $X^G$ is the
type of actual fixed points,
and the projection to $X^{hG}$ implements the map from actual fixed
points to homotopy fixed points.

Even without going to the orbit category, we can say something about
topological groups through their classifying types in type theory. For
example~\cite{Camarena}, if $f : H \to G$ is injective, then the
homotopy fiber of $Bf$ is by Corollary~\ref{cor:hfiber-hom}
is the homotopy orbit space $G \dblslash H$, which in this case is
just the coset space $G/H$, and hence in type
theory represents the homotopy type of this coset space. And if
\[
  1 \to K \to G \to H \to 1
\]
is a short exact sequence of topological groups,
then $BK \to BG \to BH$ is a fibration sequence,
i.e., we can recover the delooping $BK$ of $K$ as the homotopy fiber
of the map $BG \to BH$.

\subsection{Some elementary constructions}
\label{sec:elementary-constructions}

If we are given a homomorphism $\varphi : H \to \Aut(N)$, represented
by a pointed map
$B\varphi : BH \to_\pt \BAut_\pt(BN)$ where $\BAut_\pt(BN)$ is the
type of pointed types merely equivalent to $BN$,
we can build a new group, the
\emph{semidirect product}, $G := H \ltimes_\varphi N$
with classifying type $BG := (z : BH) \times (B\varphi\,z)$.
The type $BG$ is indeed pointed (by the pair of the basepoint $\pt$ in $BG$
and the basepoint in the pointed type $B\varphi(\pt)$), and
connected, and hence presents a higher group $G$.
An element of $g$ is given by a pair of an element $h : H$ and an
identification $g\cdot\pt = \pt$ in $B\varphi(\pt) \equiv_\pt BN$. But
since the action is via pointed maps, the second component is
equivalently an identification $\pt = \pt$ in $BN$, i.e., an element of
$N$. Under this equivalence, the product of $(h,n)$ and $(h',n')$ is
indeed $(h\cdot h', n\cdot \varphi(h)(n'))$.

As a special case we obtain the \emph{direct product} when $\varphi$
is the trivial action. Here, $B(H \times N) \equiv BH \times BN$.

As another special case we obtain the \emph{wreath products} $N \wr
\Sym_n$ of a group $N$ and a symmetric group $\Sym_n$. Here,
$\Sym_n$ acts on the direct power $N^{\Fin n}$ by permuting the
factors. Indeed, using the representation of $\BS_n$ as the type of
$n$-element types, the map $B\varphi$ is simply $A \mapsto (A \to
BN)$. Hence the delooping of the wreath product $G := N \wr \Sym_n$
is just $BG:=(A:BS_n) \times (A \to BN)$.

\section{Set-level groups}
\label{sec:n=0}
In this section we give a proof that the $n=0$ column of~\autoref{tab:periodic} is correct. 
Note that for $n=0$ the hom-types $\hom_{(0,k)}(G,H)$ are sets, which means that $(0,k)\GType$ forms a 1-category. % (F) should we define precategory/category here? I use the word "category" for a not necessarily univalent category.
Let $\Group$ be the category of ordinary set-level groups (a set with multiplication, inverse and unit satisfying the group laws) and $\AbGroup$ the category of abelian groups.
\begin{theorem}
  We have the following equivalences of categories
  {\normalfont(}\/for $k\geq2${\normalfont):}
  \begin{align*}
    % (0,0)\GType &\equiv \Set_\pt\\
    (0,1)\GType &\equiv \Group;&\text{\leanref{cGType\_equivalence\_Grp}}&\\
    (0,k)\GType &\equiv \AbGroup.&\text{\leanref{cGType\_equivalence\_AbGrp}}&
  \end{align*}
\end{theorem}
Since this theorem has been formalized we will not give all details of the proof.
\begin{proof}
  Let $k\ge1$ and $G$ be a group which is abelian if $k>1$ and let $X:\Type_\pt^{\ge k,\le k}$. If we have a group homomorphism $\phi : G \to \Omega^k X$ we get a map $e_\phi^k:K(G,k)\to_\pt X$. For $k=1$ this follows directly from the induction principle of $K(G,1)$. For $k>1$ we can define the group homomorphism $\widetilde\phi$ as the composite $G \xrightarrow{\phi} \Omega^k X \equiv \Omega^{k-1}(\Omega X)$, and apply the induction hypothesis to get a map
  $e_{\widetilde\phi}^{k-1}:K(G,k-1)\to_\pt \Omega X$. By the adjunction $\Sigma\dashv\Omega$ we get a pointed map $\Sigma K(G,k-1)\to_\pt X$, and by the elimination principle of the truncation we get a map $K(G,k)=\trunc{\Sigma K(G,k-1)}_{k}\to_\pt X$. 

  We can now show that $\Omega^k e_\phi^k$ is the expected map, that is, the following diagram commutes, but we omit this proof here.
  \begin{center}
    \begin{tikzpicture}[auto,node distance=2cm,
        thick,main node/.style={font=\sffamily\bfseries},text height=1.5ex]
      \node[main node] (OK) at (0,0) {$\Omega^nK(G,k)$};
      \node[main node] (G)  at (3,0) {$G$};
      \node[main node] (OX) at (1.5,-1.5) {$\Omega^n X$};
           \path[every node/.style={font=\sffamily\small}]
           (OK) edge [->] node {$\sim$} (G)
                edge [->] node [below left] {$\Omega^n e_\phi^k$} (OX)
           (OX) edge [->] node [below right] {$\phi$} (G);
    \end{tikzpicture}
  \end{center}
  Now if $\phi$ is a group isomorphism, by Whitehead's Theorem for truncated types \cite[Theorem 8.8.3]{TheBook} we know that $e_\phi^k$ is an equivalence, since it induces an equivalence on all homotopy groups (trivially on the levels other than $k$). We can also show that $e_\phi^k$ is natural in $\phi$.

  Note that if we have a group homomorphism $\psi:G\to G'$, we also get a group homomorphism $G\to\Omega^k K(G',k)$, and by the above construction we get a pointed map $K(\psi,k):K(G,k)\to_\pt K(G',k)$. This is functorial, which follows from naturality of $e_\phi^k$. 

  Finally, we can construct the equivalence explicitly. We have a functor
  $\pi_k:(0,k)\GType \to \AbGroup$ which sends $G$ to $\pi_k BG$. Conversely, we have the functor $K({-},k):\AbGroup\to (0,k)\GType$. We have natural isomorphisms
  $\pi_k K(G,k)\equiv G$ by~\autoref{thm:EM-spaces} and $K(\pi_k X,k)\equiv_\pt X$ by the application of Whitehead described above. The construction is exactly the same for $k=1$ after replacing $\AbGroup$ by $\Group$.
\end{proof}

\section{Stabilization}
\label{sec:stabilization}

In this section we discuss some constructions with higher groups~\cite{BaezDolan1998}. We will give the actions on the carriers and the deloopings, but we omit the third component, the pointed equivalence, for readability. We recommend keeping \autoref{tab:periodic} in mind during these constructions.
\begin{description}
\item[decategorification] $\Decat : (n,k)\GType \to (n-1,k)\GType$\\
  $\angled{G,B^kG} \mapsto \angled{\trunc G_{n-1}, \trunc{B^kG}_{n+k-1}}$
\item[discrete categorification] $\Disc : (n,k)\GType \to (n+1,k)\GType$ \\
  $\angled{G,B^kG} \mapsto \angled{G,B^kG}$
\end{description}
These functors make $(n,k)\GType$ a reflective sub-$(\infty,1)$-category of $(n+1,k)\GType$. That is, there is an adjunction ${\Decat} \dashv {\Disc}$ \leanref{Decat\_adjoint\_Disc}\footnote{In the formalization the naturality of the adjunction is a separate statement, \leanref{Decat\_adjoint\_Disc\_natural}. This is also true for the other adjunctions.} such that the counit induces an isomorphism ${\Decat} \circ {\Disc} = \id$ \leanref{Decat\_Disc}. These properties are straightforward consequences of the universal property of truncation.

There are also iterated versions of these functors.
\begin{description}
  \item[$\infty$-decategorification] $\iDecat : (\infty,k)\GType \to (n,k)\GType$\\
    $\angled{G,B^kG} \mapsto \angled{\trunc G_{n}, \trunc{B^kG}_{n+k}}$
  \item[discrete $\infty$-categorification] $\iDisc : (n,k)\GType \to (\infty,k)\GType$ \\
    $\angled{G,B^kG} \mapsto \angled{G,B^kG}$
\end{description}
These functors satisfy the same properties: ${\iDecat} \dashv {\iDisc}$ \leanref{InfDecat\_adjoint\_InfDisc} such that the counit induces an isomorphism ${\iDecat} \circ {\iDisc} = \id$ \leanref{InfDecat\_InfDisc}.

For the next constructions, we need the following properties.
\begin{definition}
  For $A : \Type_\pt$ we define the \emph{$n$-connected cover} of $A$ to be 
  $A{\angled n} := \fib(A \to \trunc A_n)$. We have the projection $p_1: A{\angled n} \to_\pt A$.
\end{definition}
\begin{lemma} \label{lem:connected-cover-univ}
  The universal property of the $n$-connected cover states the following. For any $n$-connected pointed type $B$, the pointed map
  $$(B \to_\pt A{\angled n}) \to_\pt (B \to_\pt A),$$
  given by postcomposition with $p_1$, is an equivalence.
  \leanref{connect\_intro\_pequiv}
\end{lemma}
\begin{proof}
  Given a map $f:B\to_\pt A$, we can form a map $\widetilde f: B \to A{\angled n}$. First note that for $b:B$ the type $\tr{fb}_n=_{\trunc A_n}\tr{\pt}_n$ is $(n-1)$-truncated and inhabited for $b=\pt$. Since $B$ is $n$-connected, the universal property for connected types shows that we can construct a $qb:\tr{fb}_n=\tr{\pt}_n$ for all $b$ such that $q_0:qb_0\cdot\ap_{\tr{\blank}_n}(f_0)=1$. Then we can define the map $\widetilde f(b):=(fb, qb)$. Now $\widetilde f$ is pointed, because $(f_0,q_0):(fb_0,qb_0)=(a_0,1)$.

  Now we show that this is indeed an inverse to the given map. On the one hand, we need to show that if $f: B \to_\pt A$, then $p_1 \circ \widetilde f=f$. The underlying functions are equal because they both send $b$ to $f(b)$. They respect points in the same way, because
  $\ap_{p_1}(\widetilde f_0)=f_0$. The proof that the other composite is the identity follows from a computation using fibers and connectivity, which we omit here, but can be found in the formalization.
\end{proof}
The next reflective sub-$(\infty,1)$-category is formed by looping and delooping.
\begin{description}
\item[looping] $\Omega : (n,k)\GType \to (n-1,k+1)\GType$ \\
  $\angled{G,B^kG} \mapsto \angled{\Omega G,B^kG{\angled k}}$
\item[delooping] $\B : (n,k)\GType \to (n+1,k-1)\GType$ \\
  $\angled{G,B^kG} \mapsto \angled{\Omega^{k-1}B^kG,B^kG}$
\end{description}
We have ${\B} \dashv {\Omega}$ \leanref{Deloop\_adjoint\_Loop}, which follows from Lemma \ref{lem:connected-cover-univ} %note: autoref writes "Theorem"
and $\Omega\circ{\B} = \id$ \leanref{Loop\_Deloop}, which follows from the fact that $A{\angled n}=A$ if $A$ is $n$-connected.

The last adjoint pair of functors is given by stabilization and forgetting. This does not form a reflective sub-$(\infty,1)$-category.
\begin{description}
\item[forgetting] $F : (n,k)\GType \to (n,k-1)\GType$ \\
  $\angled{G,B^kG} \mapsto \angled{G,\Omega B^kG}$
\item[stabilization] $S : (n,k)\GType \to (n,k+1)\GType$ \\
  $\angled{G,B^kG} \mapsto \angled{SG,\trunc{\susp B^kG}_{n+k+1}}$,\\
  where $SG = \trunc{\Omega^{k+1}\susp B^kG}_n$
\end{description}
We have the adjunction ${S} \dashv {F}$ \leanref{Stabilize\_adjoint\_Forget} which follows from the suspension-loop adjunction $\Sigma\dashv\Omega$ on pointed types.

The next main goal in this section is the stabilization theorem,
stating that the ditto marks in~\autoref{tab:periodic} are justified.

The following corollary is almost \cite[Lemma~8.6.2]{TheBook}, but
proving this in Book HoTT is a bit tricky. See the
formalization for details.
\begin{lemma}[Wedge connectivity]
  \label{lem:wedge-connectivity}
  If $A : \Type_\pt$ is $n$-connected and $B: \Type_\pt$ is
  $m$-connected, then the map $A \vee B \to A \times B$ is
  $(n+m)$-connected. \leanref{is\_conn\_fun\_prod\_of\_wedge}
\end{lemma}

Let us mention that there is an alternative way to prove the wedge
connectivity lemma: Recall that if $A$ is $n$-connected and $B$ is
$m$-connected, then $A \ast B$ is
$(n+m+2)$-connected~\cite[Theorem~6.8]{joinconstruction}. Hence the
wedge connectivity lemma is also a direct consequence of the following lemma.
\begin{lemma}
Let $A$ and $B$ be pointed types.
The fiber of the wedge inclusion $A\vee B\to A\times B$ is equivalent to
$\Omega{A}\ast\Omega{B}$. 
\end{lemma}
\begin{proof}
Note that the fiber of $A\to A\times B$ is $\Omega B$, the fiber of $B\to A\times B$ is $\Omega A$, and of course the fiber of $1\to A\times B$ is $\Omega A\times \Omega B$. We get a commuting cube
\begin{equation*}
\begin{tikzcd}
& \Omega A\times \Omega B \arrow[dl] \arrow[d] \arrow[dr] \\
\Omega B \arrow[d] & 1 \arrow[dl] \arrow[dr] & \Omega A \arrow[dl,crossing over] \arrow[d] \\
A \arrow[dr] & 1 \arrow[d] \arrow[from=ul,crossing over] & B \arrow[dl] \\
& A\times B
\end{tikzcd}
\end{equation*}
in which the vertical squares are pullback squares. 

By the descent theorem for pushouts it now follows that $\Omega A\ast \Omega B$ is the fiber of the wedge inclusion.
\end{proof}

The second main tool we need for the stabilization theorem is:
\begin{theorem}[Freudenthal]
  If $A : \Type_\pt^{>n}$ with $n\ge 0$, then the map
  $A \to \Omega\susp A$ is $2n$-connected.
\end{theorem}
This is \cite[Theorem~8.6.4]{TheBook}.

The final building block we need is:
\begin{lemma}
  There is a pullback square
  \[
    \begin{tikzcd}
      \susp\Omega A \ar[d,"\varepsilon_A"']\ar[r] & A \vee A \ar[d] \\
      A \ar[r,"\Delta"'] & A \times A
    \end{tikzcd}
  \]
  for any $A : \Type_\pt$.
\end{lemma}

\begin{proof}
Note that the pullback of $\Delta:A\to A\times A$ along either inclusion $A\to A\times A$ is contractible. So we have a cube
\begin{equation*}
\begin{tikzcd}
& \Omega A \arrow[dl] \arrow[d] \arrow[dr] \\
1 \arrow[d] & 1 \arrow[dl] \arrow[dr] & 1 \arrow[dl,crossing over] \arrow[d] \\
A \arrow[dr] & A \arrow[d,"\Delta"] \arrow[from=ul,crossing over] & A \arrow[dl] \\
& A\times A
\end{tikzcd}
\end{equation*}
in which the vertical squares are all pullback squares. Therefore, if we pull back along the wedge inclusion, we obtain by the descent theorem for pushouts that the square in the statement is indeed a pullback square.
\end{proof}

\begin{theorem}[Stabilization]
  \label{thm:stabilization}
  If $k\ge n+2$, then $S : (n,k)\GType \to (n,k+1)\GType$ is an
  equivalence, and any $G : (n,k)\GType$ is an infinite loop space. \leanref{stabilization}
\end{theorem}
\begin{proof}
  We show that $F\circ S=\id=S\circ F : (n,k)\GType \to (n,k)\GType$
  whenever $k\ge n+2$.

  For the first, the unit map of the adjunction factors as
  \[
    B^kG \to \Omega\susp B^kG \to \Omega\trunc{\susp B^kG}_{n+k+1}
  \]
  where the first map is $2k-2$-connected by Freudenthal, and the
  second map is $n+k$-connected. Since the domain is $n+k$-truncated,
  the composite is an equivalence whenever $2k-2 \ge n+k$.

  For the second, the counit map of the adjunction factors as
  \[
    \trunc{\susp\Omega B^kG}_{n+k} \to \trunc{B^kG}_{n+k} \to B^kG,
  \]
  where the second map is an equivalence. By the two lemmas above, the
  first map is $2k-2$-connected.
\end{proof}
For example, for $G : (0,2)\GType$ an abelian group, we have
$B^nG = K(G,n)$, an Eilenberg-MacLane space.

The adjunction ${S} \dashv {F}$ implies that the free group on a
pointed set $X$ is $\Omega\trunc{\susp X}_1=\pi_1(\susp X)$.  If $X$
has decidable equality, $\susp X$ is already $1$-truncated. It is an
open problem whether this is true in general.

Also, the abelianization of a set-level group $G : 1\Grp$ is
$\pi_2(\susp BG)$. If $G : (n,k)\GType$ is in the stable range ($k \ge
n+2$), then $SFG=G$.

\section{Perspectives on ordinary group theory}
\label{sec:perspectives}

In this section we shall indicate how the theory of higher groups can
yield a new perspective even on ordinary group theory.

From the symmetric groups $\Sym_n$, we can get other finite groups using
the constructions of~\autoref{sec:elementary-constructions}. Other
groups can be constructed more directly. For example,
$BA_n$, the classifying type of the alternating group, can be taken to
be the type of $n$-element sets $X$ equipped with a \emph{sign
  ordering}: this is an equivalence class of an ordering
$\Fin n \equiv X$ modulo even permutations. Indeed, there are only two
possible sign orderings, so this definition corresponds to
first considering the short exact sequence
\[
  1 \to A_n \to \Sym_n \xrightarrow{\mathrm{sgn}}{} \Sym_2 \to 1
\]
where the last map is the sign map, then realizing the sign map
as given by the map $\mathrm{Bsgn} : \BS_n \to \BS_2$ that takes
an $n$-element set to its set of sign orderings, and finally
letting $BA_n$ be the homotopy fiber of $\mathrm{Bsgn}$.

Similarly, $BC_n$, the classifying type of the cyclic group on $n$
elements, can be taken to be the type of $n$-elements sets $X$
equipped with a \emph{cyclic ordering}: an equivalence class of an
ordering $\Fin n \equiv X$ modulo cyclic permutations. But unlike
the above, where we had the coincidence that $\Aut(\Sym_2) \equiv
\Sym_2$,
this doesn't corresponds to a short exact sequence. Rather,
it corresponds to a sequence
\[
  1 \to C_n \to \Sym_n \to \Aut(\Fin(n-1)) \equiv \Sym_{(n-1)!}
\]
where the delooping of the last map is the map from $\BS_n$ to
$\BS_{(n-1)!}$ that maps an $n$-element set to the set of cyclic
orderings, of which there are $(n-1)!$ many -- since once we fix the
position in the ordering of a particular element,
we are free to permute the rest.

As another example,
consider the map $p : \BS_4 \to_\pt \BS_3$ that maps a 4-element set
$X$ to its set of 2-by-2 partitions, of which there $3$. Using this
construction, we can realize some famous semidirect and wreath product identities,
such as $A_4 \equiv S_2^2 \rtimes A_3$, $S_4 \equiv S_2^2 \rtimes
S_3$, and, for the octahedral group, $O_h \equiv S_2^3 \rtimes S_3
\equiv S_2 \wr S_3$.

\smallskip

Let us turn to a different way of getting new groups from old, namely
via covering space theory.

\subsection{\texorpdfstring{$1$}{1}-groups and covering spaces}
\label{sec:covering}

The connections between covering spaces of a pointed connected type
$X$ and sets with an action of the fundamental group of $X$ has
already been established in homotopy type
theory~\cite{FavoniaHarper2016}. Let us recall this connection and
expand a bit upon it.

For us, a pointed connected type $X$ is equivalently
an $\infty$-group $G:\infty\Grp$
with delooping $BG := X$.
A covering space over $BG$ is simply a type family $C : BG \to \Set$
that lands in the universe of sets.
Hence by our discussion of actions in~\autoref{sec:actions}
it is precisely a set with a $G$-action.
Since $\Set$ is a 1-type, $C$ extends uniquely to a type family
$C' : \trunc{BG}_1 \to \Set$,
but $\trunc{BG}_1$ is the delooping of the fundamental group
of $X$, and hence $C'$ is the uniquely determined
choice of a set with an action of the fundamental group.

The universal covering space is the simply connected cover of $BG$,
\[
  \widetilde{BG} : BG \to \Set, \quad
  z \mapsto \trunc{\pt = z}_0.
\]
Note that the total space of $\widetilde{BG}$ is indeed the
$1$-connected cover $BG\angled1$,
since $\trunc{\pt =_{BG} \pt}_0 \equiv (\tr{\pt} =_{\trunc{BG}_1} \tr{\pt})$.
Also note that if $G$ is already a 1-group, then this is just the right
action of $G$ on itself, and in general, it is the right action of $G$
on the fundamental group
(i.e., the decategorification of $G$)
via the truncation homomorphism from $G$ to $\pi_1(BG)$,
where we can also view $\pi_1(BG)$ as the $1\Grp$ decategorification
of $G$.

In general, there is a Galois correspondence between connected covers
of $BG$ and conjugacy class of subgroups of the fundamental group.
Indeed, if $C : BG \to \Set$ has a connected total space,
then the space $(g : \trunc{BG}_1) \times C'(g)$
is itself a connected, 1-truncated type,
and the projection to $\trunc{BG}_1$
induced an inclusion of fundamental groups
once a point $\pt : C'(\pt)$ has been chosen.

\begin{theorem}[Fundamental theorem of Galois theory for covering spaces]
  \label{thm:galois-zero}
  $\phantom{42}$
  \begin{enumerate}
  \item  The automorphism group of the universal covering space
    $\widetilde{BG}$ is isomorphic to
    the $1$-group decategorification of $G$,
    \[
      \Aut(\widetilde{BG}) \equiv \Decat_1(G) \equiv \pi_1(BG).
    \]
  \item Furthermore, there is an contravariant correspondence between
    conjugacy classes of subgroups of $\Decat_1(G)$ and connected
    covers of $BG$.
  \item This lifts to a Galois correspondence between subgroups of
    $\Decat_1(G)$ and pointed, connected covers of $BG$.  The normal
    subgroups correspond to Galois covers.
  \end{enumerate}
\end{theorem}
Note that the universal covering space
and the trivial covering space
(constant at the unit type)
are canonically pointed,
reflecting the fact that
the two trivial subgroups are normal.

The first part of the fundamental theorem has a clear generalization to
higher groups:
\begin{theorem}[Fundamental theorem of Galois theory for $n$-covers,
  part one]
  The automorphism group of the universal $n$-type cover $U_n(BG)$,
  \[
    U_n(BG) : BG \to \Type^{\le n},
    \quad
    z \mapsto \trunc{\pt = z}_n
  \]
  of $BG$ is
  isomorphic to the $(n+1)$-group decategorification of $G$,
  \[
    \Aut(U_n(BG)) \equiv \Decat_{n+1}(G) \equiv \Pi_{n+1}(BG).
  \]
\end{theorem}
\begin{proof}
  Note that
  $\BAut(U_n(BG))$ is the image of the map $1 \to (BG \to \Type^{\le
    n})$ that sends the canonical element to $U_n(BG)$. Since $BG$ is
  connected, this image is exactly $\trunc{BG}_{n+1}$ by
  \cite[Theorem~7.1]{joinconstruction}. Then we are done,
  since $\B\Pi_{n+1}(BG) \equiv \trunc{BG}_{n+1}$, by definition.
\end{proof}
It is possible to use the other parts
of~\autoref{thm:galois-zero} in order to \emph{define} the notions of
subgroup and normal subgroup for $n$-groups, which then become
\emph{structure on} rather than a \emph{property of} a homomorphism $f
: K \to G$.
Explicitly, the structure of a \emph{normal subgroup} on such an $f$
is a delooping $B(G \dblslash K)$ of the type $G \dblslash K$
together with a map $Bq : BG \to_\pt B(G \dblslash K)$ giving rise to a
fiber sequence
\begin{equation}\label{eq:normal-fiber-sequence}
  G \dblslash K \to BK \xrightarrow{Bf}{} BG
  \xrightarrow{Bq}{} B(G \dblslash K).
\end{equation}

\subsection{Central extensions and group cohomology}
\label{sec:group-cohomology}

The cohomology of a higher group $G$ is simply the cohomology of its
delooping $BG$. Indeed, for any spectrum $A$, we define
\[
  H_{\mathrm{Grp}}^k(G, A) := \trunc{BG \to_\pt B^kA}_0.
\]
Of course, to define the $k$'th cohomology group, we only need the
$k$-fold delooping $B^kA$.

If $A:(\infty,2)\GType$ is a braided $\infty$-group, then
we have the second cohomology group $H_{\mathrm{Grp}}^2(G, A)$, and an
element $c:BG \to_\pt B^2A$ gives rise to a \emph{central extension}
\[
  BA \to BH \to BG \xrightarrow{c}{} B^2A,
\]
where $BH$ is the homotopy fiber of $c$.  This lifts to the world of
higher groups the usual result that isomorphism classes of central
extensions of a $1$-group $G$ by an abelian $1$-group $A$
are given by cohomology classes in $H_{\mathrm{Grp}}^2(G, A)$.

\smallskip

In the Spectral repository there is full formalization of the Serre
spectral sequence for cohomology \cite{SerreSpectralSequence}.
If we have any normal subgroup fiber sequence for $\infty$-groups as
in~\eqref{eq:normal-fiber-sequence}, then we get a corresponding
spectral sequence with $E_2$-page
\[
  H_{\mathrm{Grp}}^p(G \dblslash K, H_{\mathrm{Grp}}^q(K, A))
\]
and converging to $H_{\mathrm{Grp}}^n(G, A)$, where $A$ is any
truncated, connective spectrum, which could even be a left $G$-module,
in which case we reproduce the \emph{Hochschild-Serre spectral
  sequence}.

% Returning to the cyclic groups $C_n$, we can also see these as
% quotients of the normal subgroups $n\bZ$ of $\bZ$. Hence they are
% automorphism groups of Galois covers of the circle, namely the
% $n$-sheeted covers where the generator of $\bZ$ acts transitively.
% More precisely, we have the homomorphism $\bZ \to \Sym_n$ that maps
% the generator to the permutation $\{0,\dots,n-1\} \ni k \mapsto
% k+1\mod n \in \{0,\dots,n-1\}$. 
% 
% Putting everything together, we now want to explain how 
% 
% [TODO: discuss fiber sequences and connection to LES of homotopy groups.
% 
% Information about group cohomology via the spectral sequence (plus reference!)]
% 
% Fibration sequence $\BS_3 \to Type$ with fiber $\BS_2^2=(\RP^\infty)^2$. Last map takes a
% four element set to its set of 2 by 2 partitions.
% 
% 
% [TODO:]
% Other finite groups using covering space theory.
% 
% Consider the map $p : \BS_4 \to_\pt \BS_3$ that maps a 4-element set
% $X$ to its set of 2-by-2 partitions,
% and pointed via the identification
% \[
%   \Bigl\{
%   \bigl\{\{1,2\},\{3,4\}\bigr\},
%   \bigl\{\{1,3\},\{2,4\}\bigr\},
%   \bigl\{\{1,4\},\{2,3\}\bigr\}
%   \Bigr\}
%   \equiv \{1,2,3\},
% \]
% using the displayed order.
% This defines a homomorphism from $\Sigma_4$ to $\Sigma_3$
% whose kernel is the Klein Vierer-group $V:=(\bZ/2)^2$.
% (Note that the pointed above is arbitrary,
% while the underlying map is not.)
% \begin{align*}
%   1 &\to (\bZ/2)^2 \to A_4 \to \bZ/3 \to 1,\\
%   1 &\to (\bZ/2)^2 \to \Sym_4 \to \Sym_3 \to 1.
% \end{align*}

\section{Formalization}
\label{sec:formalization}
We have formalized many results of this paper. We use the proof assistant Lean 2\footnote{\url{https://github.com/leanprover/lean2}}. This is an older version of the proof assistant Lean\footnote{\url{https://leanprover.github.io/}} (version 3.3 as of January 2018). We use the old version, since the newer version doesn't officially support HoTT, although there is an experimental library for HoTT\footnote{\url{https://github.com/gebner/hott3}}, but that doesn't have as much theory as the library in Lean 2.

The Lean 2 HoTT library is divided in two parts, the core library%
\footnote{\url{https://github.com/leanprover/lean2/blob/master/hott/hott.md}}
and the formalization of spectral
sequences\footnote{\url{https://github.com/cmu-phil/Spectral}}. We
worked in the latter, so that we could use the results from that
repository, such as theorems about Eilenberg-MacLane spaces and
pointed maps. All results in this paper are stated in one file%
\footnote{\url{https://github.com/cmu-phil/Spectral/blob/master/higher_groups.hlean}}, although for many results the main parts of the proof is elsewhere (in Emacs, click on a name and press \texttt{M-.} to find a definition).

To build the file, install Lean 2 via the instructions from that
repository, and then download the Spectral repository and compile it
(you can use the command \texttt{path/to/lean2/bin/linja} on the
command-line to compile the library you're in). The Spectral
repository contains some unproven results, marked by \textbf{sorry}. You can write \texttt{print axioms theoremname} in a file to ensure that \textbf{sorry} isn't used in the proof.

\section{Conclusion}
\label{sec:conclusion}

We have presented a theory and formalization of higher groups in
HoTT, and we have proved that for set-level structures we recover the
well-known objects: groups and abelian groups. A possible next step
would be to do the same for the $1$-type objects.
The corresponding algebraic objects have a long history.
Strict 2-groups predate category theory as they originate in
Whitehead's study of \emph{crossed modules} \cite{Whitehead1949}.
The theory of weak $2$-groups was begun by Grothendieck's student Hoàng
Xuân Sính \cite{Sinh1975} and further developed in \cite{BaezLauda2004}.
It should be possible to prove within HoTT that weak $2$-groups and
crossed modules are equivalent to $2$-groups in our sense, when we use the
respective, correct notions of equivalence.

Symmetric $2$-groups are by the stabilization theorem the same as
$1$-truncated symmetric spectra. These are described more simply than
arbitrary crossed modules as \emph{Picard groupoids}. This is part of the
stable homotopy hypothesis~\cite{Gurski-stable-homotopy-hypothesis,JohnsonOsorno}.
It should also be possible to develop the theory of Picard groupoids
in HoTT, and thus prove the corresponding stable homotopy hypothesis.

Higher groups have been intensively studied in homotopy theory, in
particular after $p$-completion for $p$ a prime. A \emph{p-compact
  group} is an $\bF_p$-local $\infty$-group whose carrier is
$\bF_p$-finite, see \cite{DwyerWilkerson1994}. They are good
homotopical analogues of Lie groups, and they interact nicely with
compact Lie groups, for instance:
\begin{theorem}[\cite{DwyerZabrodsky1987}]
  Let $P$ be a $p$-toral group, and let $G$ be a compact Lie group. Then
  $\trunc{BP\to_\pt BG}_0$ is isomorphic to the conjugacy classes of
  homomorphisms from $P$ to $G$.
\end{theorem}

Higher groups also play a particularly prominent role in the
development of quantum field theory in cohesive homotopy type
theory~\cite{ShreiberShulman}. In cohesive type theory we can actually
capture the topological or smooth structure of groups and their
classifying types, and hence develop Lie theory properly, including
the higher group generalization thereof.
All of our results only use the core part of HoTT, and hence they
remain valid also in cohesive HoTT.

Note that we have crucially used a trick to study higher groups in
HoTT, namely that these can be represented by pointed, connected
types. The alternative would have been to define them as group-like
algebras for the little $k$-cubes operad $E_k$. 
But this requires exactly the kind of infinitary tower of coherence
conditions that we don't yet know how to define in HoTT. (Or whether it
is even possible.) Thus, while
we have the type of higher groups, we do not have the type of higher
monoids (general $E_k$-algebras).
Thus their theory, and the corresponding stabilization theorem, is
currently beyond the reach of HoTT.
% Mention two-level type theory, computational type theory?
% Spectrification will be covered elsewhere(?)

\section*{Acknowledgement}
  The authors gratefully acknowledge the support of the Air Force
  Office of Scientific Research through MURI grant
  FA9550-15-1-0053. Any opinions, findings and conclusions or
  recommendations expressed in this material are those of the authors
  and do not necessarily reflect the views of the AFOSR.
    
\printbibliography
\end{document}